\documentclass[a4paper,11pt]{article}

\usepackage{amsmath,amssymb,amsthm}
\usepackage[english]{babel}
\usepackage{hyperref}

\usepackage[margin=2.9cm]{geometry}

\theoremstyle{definition}
\newtheorem{theorem}{Theorem}[section]
\newtheorem{proposition}[theorem]{Proposition}
\newtheorem{lemma}[theorem]{Lemma}
\newtheorem{corollary}[theorem]{Corollary}
\newtheorem{notation}[theorem]{Notation}
\newtheorem{definition}[theorem]{Definition}
\newtheorem{example}[theorem]{Example}
\newtheorem{remark}[theorem]{Remark}

\newcommand{\supp}{\mathop{\text{supp}}}

\newcommand{\rk}{\mathop{\text{rk}}}
\newcommand{\Mat}{\mathop{\text{Mat}}}
\newcommand{\Tr}{\mathop{\text{Tr}}}
\newcommand{\colsp}{\mathop{\text{colsp}}}
\newcommand{\rowsp}{\mathop{\text{rowsp}}}
\newcommand{\F}{\mathbb{F}}
\newcommand{\mC}{\mathcal{C}}
\newcommand{\mA}{\mathcal{A}}
\newcommand{\mD}{\mathcal{D}}
\newcommand{\mB}{\mathcal{B}}

\newcommand{\mAA}{\textnormal{Ant}}
\newcommand{\cc}{\textnormal{c}}
\newcommand{\rr}{\textnormal{r}}

\newcommand{\dbot}{\mathbin{\text{$\bot\mkern-10mu\bot$}}}

 \usepackage{enumitem}

\title{\textbf{Rank-Metric Codes and $q$-Polymatroids}}
\date{}

\usepackage{authblk}

\author[1]{Elisa Gorla}
\author[2]{Relinde Jurrius}
\author[3]{Hiram H. L\'opez}
\author[4]{Alberto Ravagnani}

\affil[1]{Institut de Mat\'ematiques, Universit\'e de Neuch\^atel, Switzerland}
\affil[2]{Faculty of Military Science, Netherlands Defence Academy, The Netherlands}
\affil[3]{Department of Mathematics, Cleveland State University, USA}
\affil[4]{School of Mathematics \& Statistics, University College Dublin, Ireland}

\begin{document}

\maketitle

\begin{abstract}
This paper contributes to the study of rank-metric codes from an algebraic and combinatorial point of view. We introduce $q$-polymatroids, the $q$-analogue of polymatroids, and develop their basic properties. We associate a pair of $q$-polymatroids to a rank-metric codes and show that several invariants and structural properties of the code, such as generalized weights, the property of being MRD or an optimal anticode, and duality, are captured by the associated combinatorial object.
\end{abstract}

\bigskip

\bigskip


\section*{Introduction and Motivation}

Rank-metric codes were originally introduced by Delsarte~\cite{del1} and later rediscovered by Gabidulin~\cite{gabid} and Roth~\cite{roth}. Due to their application in network coding, the interest in these codes has intensified over the past years and many recent papers have been devoted to their study. While interest in these codes stems from practical applications, rank-metric codes also present interesting algebraic and combinatorial properties. Therefore, their mathematical structure has also been the object of several works. This paper belongs to the latter line of study. 
Our contributions are twofold: on the one side we study generalized weights of rank-metric codes, and on the other we establish a link with other combinatorial objects. More precisely, we associate to each rank-metric code a $q$-polymatroid, the $q$-analogue of a polymatroid, which we define here.

In Section~\ref{secRankMet} we define rank-metric codes and vector rank-metric codes. We recall how to associate a rank-metric code to a vector rank-metric code via the choice of a basis and establish a number of basic, but fundamental facts. In particular, we recall the notions of equivalence for rank-metric codes and vector rank-metric codes and we discuss in detail why these notions are compatible via the association mentioned above. We also explain that, while the choice of a basis affects the rank-metric code obtained via the association above, the equivalence class of the rank-metric code obtained does not depend on the choice of the basis.

Generalized weights have been defined and studied in different levels of generality by many researchers. Two of the first definitions of generalized weights for vector rank-metric codes are due to Kurihara, Matsumoto, and Uyematsu~\cite{KMU} and to Oggier and Sboui~\cite {OS}. More definitions are due to Jurrius and Pellikaan~\cite{Ondef} and Mart{\'\i}nez-Pe{\~n}as and Matsumoto~\cite{U6}, who also compared the various definitions. Using the theory of anticodes, Ravagnani~\cite{albgen} gave a definition of generalized weights for matrix rank-metric codes, which extends the one from~\cite{KMU}.

In this paper, we develop further the theory of generalized weights for rank-metric codes, tying together several previously known results on the subject. We adopt the definition from~\cite{albgen} and, in Section~\ref{anticodesweights}, we show that it is invariant with respect to equivalence of rank-metric codes. We also show that the definition of generalized weights for rank-metric codes from~\cite{U6}, which generalizes definitions from~\cite{Klove, Wei, Ondef}, is not invariant with respect to code equivalence. 

Given the well known link between codes in the Hamming metric and matroids, it is a natural question to ask if there is a $q$-analogue of this. Rank-metric codes can be viewed as the $q$-analogue of block-codes endowed with the Hamming metric. So it is natural to ask what is the $q$-analogue of a matroid. Crapo~\cite{Cr64} already studied this combinatorial object from the point of view of geometric lattices. Recently, Jurrius and Pellikaan~\cite{JP16} rediscovered $q$-matroids and associated a $q$-matroid to every vector rank-metric code. One goal of the current paper is extending this association to rank-metric codes.

With this in mind, we define the $q$-analogue of a polymatroid, that we call a $q$-polymatroid. In Section~\ref{secqPolym} we develop basic properties of $q$-polymatroids, such as equivalence and duality. In Section~\ref{secCodesPolym} we associate to every rank-metric code a pair of $q$-polymatroids. We also show that the $q$-polymatroids arising from rank-metric codes are in general not $q$-matroids. We then show that several structural properties of rank-metric codes depend only on the associated $q$-polymatroid: In Section~\ref{secStr} we do this for the minimum distance and the property of being MRD, in Section~\ref{secGenWtPoly} for the generalized weights and for the property of being an optimal anticode, and in Section~\ref{secDuality} for duality. These results are $q$-analogues of classical results in coding theory.

While preparing this manuscript, we became aware that a slightly different definition of $q$-analogue of a polymatroid was given independently by Shiromoto~\cite{Sh18}. While our paper applies this theory to equivalence of codes and to generalized weights,~\cite{Sh18} focuses on the weight enumerator of rank-metric codes.

\paragraph*{Acknowledgements.} 
Elisa Gorla was partially supported by the Swiss National Science Foundation through grant no. 200021\_150207.
Hiram H. L\'opez was partially supported by SNI, Mexico. Alberto Ravagnani was partially supported by the Swiss National Science Foundation through grant no. P2NEP2\_168527.

\bigskip

\bigskip

\section{Rank-Metric and Vector Rank-Metric Codes}\label{secRankMet}

We start by establishing the notation and the definitions used throughout the paper.

\begin{notation}
In the sequel, we fix integers $n,m\geq 2$ and a prime power $q$. For an integer $t$, we let $[t]:=\{1,...,t\}$.
We denote by $\F_q$ the finite field with $q$ elements. 
The space of $n\times m$ matrices with entries in $\mathbb{F}_q$ is denoted by $\Mat$. Up to transposition, we assume without loss of generality that $n\leq m$. 
We let $$\Mat(J,c)=\{M\in\Mat\mid\colsp(M)\subseteq J \} \quad\mbox{and} \quad \Mat(J,r)=\{M\in\Mat\mid\rowsp(M)\subseteq J \}.$$
\end{notation}

Throughout the paper, we only consider linear codes. All dimensions are computed over $\mathbb{F}_q$, unless otherwise stated.

\begin{definition}
A (\textbf{matrix}) \textbf{rank-metric code} is an $\F_q$-linear subspace $\mC \subseteq \Mat$. The \textbf{dual} of $\mC$ is 
$$\mathcal{C}^\perp=\{M \in\Mat \mid \Tr(MN^t)=0 \text{ for all }N\in\mathcal{C}\},$$
where $\Tr(\cdot)$ denotes the trace.
It is easy to check that $\mC^\perp$ is a code as well, i.e., that it is $\F_q$-linear. 
The \textbf{minimum} (\textbf{rank}) \textbf{distance} of a non-zero rank-metric code $\mC \subseteq \Mat$ is the integer $d(\mC):=\min\{\rk(M) \mid M \in \mC, \ M \neq 0\}$.
\end{definition}

The next bound relates the dimension of a code $\mC \subseteq \Mat$ to its minimum distance. It is the analogue for the rank metric of the Singleton bound from classical coding theory.

\begin{proposition}[\cite{del1}, Theorem~5.4]\label{Sing}
Let $\mC \subseteq \Mat$ be a non-zero rank-metric code with minimum distance $d$. Then $\dim(\mC) \le m (n-d+1)$.
\end{proposition}

\begin{definition}A code that attains the bound of Proposition~\ref{Sing} is called a \textbf{maximum rank distance} (\textbf{MRD}) code.
\end{definition}

We now introduce some transformations that preserve the dimension and the minimum rank distance of a rank-metric code.
These will play a central role throughout the paper.

\begin{notation}
Let $\mC \subseteq \Mat$ be a rank-metric code, let $A\in\mbox{GL}_n(\F_q)$ and $B\in\mbox{GL}_m(\F_q)$. Define
$$A\mC B:=\{AMB\mid M\in\mC\} \subseteq \Mat.$$
When $n=m$, define the \textbf{transpose} of a rank-metric code $\mC \subseteq \Mat$ as $$\mC^t:=\{M^t \mid M \in \mC\} \subseteq \Mat.$$
\end{notation}

As we are interested in structural properties of rank-metric codes, it is natural to study these objects up to equivalence. 
Linear isometries of the space of matrices of fixed size induce a natural notion of equivalence among rank-metric codes.

\begin{definition}
Two rank-metric codes $\mC_1,\mC_2\subseteq\Mat$ are \textbf{equivalent} if there exists an $\F_q$-linear isometry $f:\Mat \to \Mat$ such that $f(\mC_1)=\mC_2$. 
If this is the case, then we write $\mC_1 \sim \mC_2$.
\end{definition}

The next theorem gives a characterization of the linear isometries of $\Mat$. It combines results by Hua and Wan, and it can be found in the form stated below in~\cite[Theorem~3.4]{wan2}.

\begin{theorem}[\cite{hua}, \cite{wan}] \label{theoiso}
Let $f:\Mat \to \Mat$ be an $\F_q$-linear isometry with respect to the rank metric. 
\begin{enumerate}[label=(\arabic*)] \setlength\itemsep{0em}
\item If $m<n$, then there exist matrices $A \in \mbox{GL}_n(\F_q)$ and $B\in \mbox{GL}_m(\F_q)$ such that $f(M)=AMB$ for all $M \in \Mat$.
\item If $m=n$, then there exist matrices $A,B \in \mbox{GL}_m(\F_q)$ such that either $f(M)=AMB$ for all $M \in \Mat$, or
$f(M)=AM^tB$ for all $M \in \Mat$.
\end{enumerate}
\end{theorem}

A class of codes that has recently received a lot of attention is that of vector rank-metric codes, introduced independently by Gabidulin and Roth in~\cite{gabid} and~\cite{roth}, respectively.

\begin{definition}
The $\textbf{rank weight}$ $\rk(v)$ of a vector $v \in \F_{q^m}^n$ is the dimension of the $\F_q$-linear space generated by its entries. A \textbf{vector rank-metric code} is an $\F_{q^m}$-linear subspace $C \subseteq \F_{q^m}^n$.  
The \textbf{dual} of $C$ is the vector rank-metric code $$C^{\dbot}:= \{v \in \F_{q^m}^n \mid \langle v, w\rangle =0 \mbox{ for all } w \in C \},$$ where $\langle \cdot , \cdot \rangle$ is the standard inner product of $\F_{q^m}^n$. 
When $C \neq \{0\}$ is a non-zero vector rank-metric code, the \textbf{minimum} (\textbf{rank}) \textbf{distance} of $C$ is  $d(C)=\min\{\rk(v)  \mid v \in C, \ v\neq 0\}$.
\end{definition}

\begin{notation}
Let $C\subseteq\F_{q^m}^n$ be a vector rank-metric code and $B\in\mbox{GL}_n(\F_q)$. Define
$$C B:=\{vB\mid v\in C\} \subseteq \F_{q^m}^n.$$
\end{notation}

Similarly to the case of rank-metric codes, the linear isometries of $\F_{q^m}^n$ induce a notion of equivalence for vector rank-metric codes.

\begin{definition}
Two vector rank-metric codes $C_1,C_2\subseteq\F_{q^m}^n$ are \textbf{equivalent} if there exists an $\F_{q^m}$-linear isometry $f:\F_{q^m}^n\to \F_{q^m}^n$ such that $f(C_1)=C_2$. 
If this is the case, then we write $C_1 \sim C_2$.
\end{definition}

The linear isometries of $\F_{q^m}^n$ can be described as follows.

\begin{theorem}[\cite{Berger}]\label{theoisovector}
Let $f:\F_{q^m}^n \to \F_{q^m}^n$ be an $\F_{q^m}^n$-linear isometry with respect to the rank metric. Then there exist $\alpha\in\F_{q^m}^*$ and $B\in\mbox{GL}_n(\F_q)$ such that 
$f(v)=\alpha v B$ for all $v\in\F_{q^m}^n$.
\end{theorem}

There is a natural way to associate a rank-metric code $\mC$ to a vector rank-metric code $C$, in such a way that the metric properties are preserved. 
 Given an $\F_{q}$-basis $\Gamma=\{\gamma_1,...,\gamma_m\}$ of $\F_{q^m}$ and given a vector $v \in \F_{q^m}^n$, let $\Gamma(v)$ denote the unique $n \times m$ matrix with entries in $\F_q$ that satisfies
$$v_i=\sum_{j=1}^m \Gamma_{ij}(v) \gamma_j \quad \mbox{for all } 1 \le i \le n.$$

\begin{proposition}[\cite{costch}, Section 1] \label{lift}
The map $v \mapsto \Gamma(v)$ is an $\F_q$-linear isometry.  In particular,
 if $C \subseteq \F_{q^m}^n$ is a vector rank-metric code of dimension $k$ over $\F_{q^m}$, then $\Gamma(C)$ is an $\F_q$-linear rank-metric code of dimension $mk$ over $\F_q$. Moreover, if $C \neq \{0\}$, then $C$ and  $\Gamma(C)$ have the same minimum rank distance.
\end{proposition}

As one expects, the rank-metric codes obtained from equivalent vector rank-metric codes using different bases $\Gamma$ and $\Gamma'$ are equivalent.

%
%


\begin{proposition}\label{equiv}
Let $C_1,C_2\subseteq \F_{q^m}^n$ be vector rank-metric codes. Let $\Gamma$ and $\Gamma'$ be bases of $\F_{q^m}$ over $\F_q$. If $C_1\sim C_2$, then $\Gamma(C_1)\sim\Gamma'(C_2)$.
\end{proposition}

\begin{proof}
By~\cite[Lemma 27.2]{albgen}, $\Gamma(C)\sim\Gamma'(C)$. Hence
we may assume without loss of generality that $\Gamma=\Gamma'=\{\gamma_1,\ldots,\gamma_m\}$. 
By definition of $\Gamma$ 
$$\gamma_k\gamma_j=\sum_{\ell=1}^m\Gamma(\gamma_k\gamma_1,\ldots,\gamma_k\gamma_m)_{j\ell}\gamma_\ell$$
If $C_1\sim C_2$, then by Theorem~\ref{theoisovector} there exist $\alpha\in\F_{q^m}^*$ and $B=(b_{ij})\in\mbox{GL}_n(\F_q)$ such that $C_2=\alpha C_1 B$. 
If $v=(v_1,\ldots,v_n)\in\F_{q^m}^n$, then $$\alpha v_i=\sum_{h=1}^m \Gamma(v)_{ih}\alpha\gamma_h=
\sum_{j=1}^m\sum_{h=1}^m \Gamma(v)_{ih}\Gamma(\alpha\gamma_1,\ldots,\alpha\gamma_m)_{hj}\gamma_j=\sum_{j=1}^m  (\Gamma(v)\Gamma(\alpha\gamma_1,\ldots,\alpha\gamma_m))_{ij}\gamma_j.$$
Therefore $$\Gamma(\alpha v)=\Gamma(v)\Gamma(\alpha\gamma_1,\ldots,\alpha\gamma_m).$$
On the other side $$vB=\left(\sum_{k=1}^n b_{k1}v_k,\ldots,\sum_{k=1}^n b_{kn}v_k\right)=\left(\sum_{k=1}^n\sum_{j=1}^m b_{k1}\Gamma(v)_{kj}\gamma_j,\ldots,\sum_{k=1}^n\sum_{j=1}^m b_{kn}\Gamma(v)_{kj}\gamma_j\right),$$ hence $$\Gamma(vB)_{ij}=\sum_{k=1}^n b_{ki}\Gamma(v)_{kj}=(B^t\Gamma(v))_{ij},$$
that is $$\Gamma(vB)=B^t\Gamma(v).$$
Then for every $v\in C_1$ we obtain $\Gamma(\alpha vB)=B^t\Gamma(v)\Gamma(\alpha\gamma_1,\ldots,\alpha\gamma_m)$ i.e. $$\Gamma(C_2)=\Gamma(\alpha C_1 B)=B^t\Gamma(C_1)\Gamma(\alpha\gamma_1,\ldots,\alpha\gamma_m)\sim \Gamma(C_1)$$
since $B^t\in\mbox{GL}_n(\F_q)$ and $\rk(\Gamma(\alpha\gamma_1,\ldots,\alpha\gamma_m))=\rk(\alpha\gamma_1,\ldots,\alpha\gamma_m)=\rk(\gamma_1,\ldots,\gamma_m)=m$, hence
$\Gamma(\alpha\gamma_1,\ldots,\alpha\gamma_m)\in\mbox{GL}_m(\F_q)$.\end{proof}

Proposition~\ref{equiv}  suggests a natural definition of $\F_{q^m}$-linear rank-metric code in the $\F_q$-linear matrix space $\Mat$.

\begin{definition}
Let $\mC \subseteq \Mat$ be a rank-metric code. We say that $\mC$ is \textbf{$\F_{q^m}$-linear} if there exists a vector rank-metric code
$C \subseteq \F_{q^m}^n$ and a basis of $\Gamma$ of $\F_{q^m}$ over $\F_q$ such that $\mC \sim \Gamma(C)$.
\end{definition}

\section{Optimal Anticodes and Generalized Weights}\label{anticodesweights}

Optimal linear anticodes were introduced in~\cite{albgen} with the purpose of studying generalized weights in the rank metric.

\begin{definition}
The \textbf{maximum rank} of a rank-metric code $\mC \subseteq \Mat$ is $$\mbox{maxrk}(\mC):=\max \{ \rk(M) \mid M \in \mC\}.$$ A rank-metric code $\mA \subseteq \Mat$ is an
\textbf{optimal anticode} if $\dim(\mA)=m \cdot  \mbox{maxrk}(\mA)$.
\end{definition}

The class of optimal anticodes is closed with respect to duality~\cite[Theorem 54]{RR15} and code equivalence.
The properties of optimal anticodes were exploited in~\cite{albgen} to study a class of algebraic invariants of rank-metric codes, called (Delsarte) generalized weights.  

\begin{definition} \label{defalberto}
Let $\mC \subseteq \Mat$ be a non-zero code. For $i\geq 1$, the $i$-th \textbf{generalized weight} of $\mC$ is
$$a_i(\mC):= \frac{1}{m} \min\{\dim(\mA) \mid \mA \subseteq \Mat \mbox{ is an optimal anticode,} \dim(\mC \cap \mA) \ge i\}.$$
\end{definition}

\begin{remark} \label{rem:dist}
$a_1(\mC)$ is the minimum rank distance of $\mC$. See~\cite[Theorem 30]{albgen} for details.
\end{remark}

As one may expect, equivalent codes have the same generalized weights.

\begin{proposition} \label{compati}
Let $\mC_1,\mC_2 \subseteq \Mat$ be non-zero codes and assume $\mC_1 \sim \mC_2$. Then
$$a_i(\mC_1)=a_i(\mC_2) \mbox{ for every integer } i \ge 1.$$ 
\end{proposition}

\begin{proof}
Since $\mC_1 \sim \mC_2$, there exist $A \in \mbox{GL}_n(\F_q)$ and $B \in \mbox{GL}_m(\F_q)$ such that either $\mC_2=A\mC_1B$, or $\mC_2=A\mC_1^tB$ and $n=m$. 
We prove the proposition in the second case, as the proof in the first is similar.

Let $\mAA(\Mat)$ denote the set of optimal anticodes in $\Mat$, and fix a positive integer $i$. The chain of equalities
$$A(\mA \cap \mC_1)^tB = (A\mA^tB) \cap (A \mC_1^tB) =(A\mA^tB) \cap \mC_2$$ implies that 
the isometry $f: \mAA(\Mat) \to \mAA(\Mat)$ defined by $f(\mA):=A\mA^tB$ gives a bijection between the anticodes $\mA \subseteq \Mat$ such that $\dim(\mA \cap \mC_1) \ge i$ and the 
anticodes $\mB \subseteq \Mat$ such that $\dim(\mB \cap \mC_2) \ge i$. Then $\mC_1$ and $\mC_2$ have the same generalized weights by definition. 
\end{proof}

The definition of generalized weights in terms of anticodes suggests the following natural questions.
Let $\mC \subseteq \Mat$ be a rank-metric code, and let $\mA$ be an optimal anticode such that 
$\dim(\mC\cap\mA)\geq i$ and $a_i(\mC)=\dim(\mA)/m$. 
\begin{enumerate}[label=(\arabic*)]\setlength\itemsep{0em}
\item Can one find an optimal anticode $\mA^\prime$ such that $\mA\subseteq \mA^\prime$, $\dim(\mC\cap\mA^\prime)\geq i+1$, and $a_{i+1}(\mC)=\dim(\mA^\prime)/m$?
\item Can one find an optimal anticode $\mA^{\prime\prime}$ such that $\mA^{\prime\prime}\subseteq \mA$, $\dim(\mC\cap\mA^{\prime\prime})\geq i-1$, and $a_{i-1}(\mC)=\dim(\mA^{\prime\prime})/m$?
\end{enumerate}

The following example shows that the answer to both questions is negative. 
%

\begin{example}
Let $q=2$ and $n=m=3$. Let $\mC$ be the rank-metric code generated by the three independent matrices
$$M_1:=\begin{pmatrix} 1 & 0 & 0 \\ 0 & 0 & 0\\ 0 & 0 & 0\end{pmatrix}, \qquad
M_2:=\begin{pmatrix} 0 & 0 & 0  \\ 0 & 1 & 1 \\ 0 & 0 & 1\end{pmatrix}, \qquad
M_3:=\begin{pmatrix} 0 & 0 & 0  \\ 0 & 0 & 1 \\ 0 & 1 & 0\end{pmatrix}.$$
It is easy to check that $a_{1}(\mC)=1$ and $a_{2}(\mC)=2$.
By~\cite[Theorems 4 and 6]{pazzis}, the optimal anticodes in $\Mat_{3\times 3}(\F_2)$ are of the form $\Mat(J,\cc)$ or $\Mat(J,\rr)$ for some $J\subseteq\F_2^3$.
Let $\mA_1$ be an optimal anticode of dimension $3$ with $\dim(\mC\cap\mA_1)\geq 1$. 
Then we have $\mA_1=\Mat(\left<(1,0,0)\right>,\cc)$ or $\mA_1=\Mat(\left<(1,0,0)\right>,\rr)$. 
Let $\mA_2$ be an optimal anticode of dimension $6$ with $\dim(\mC\cap\mA_2)\geq 2$. Then we have $\mA_2=\Mat(\left<(0,1,0),(0,0,1)\right>,\cc)$ or $\mA_2=\Mat(\left<(0,1,0),(0,0,1)\right>,\rr).$
\end{example}

Notice that one could also define generalized weights for rank-metric codes following a support-based analogy with codes endowed with the Hamming metric. This naturally leads to generalizing the invariants proposed in~\cite{Klove}, \cite{Wei} and~\cite{Ondef} as in the following Definition~\ref{defaltra}. This approach has been followed, e.g., in~\cite{U6}. Notice that in~\cite{U6} supports are defined as column spaces also in the case when $n>m$.

\begin{definition} \label{defaltra}
Let $\mC \subseteq \Mat$ be a non-zero code. 
The \textbf{support} of a subcode $\mD \subseteq \mC$ is
$$\supp(\mD):=\sum_{M \in \mD} \colsp(M) \ \subseteq \F_q^n,$$
where the sum is the sum of vector spaces. The $i$-th \textbf{support weight} of $\mC$ is 
$$cs_i(\mC):=\min \{\dim(\supp(\mD)) \mid \mD \subseteq \mC, \ \dim(\mD)=i\}.$$
\end{definition}

\begin{remark}
Although Definition~\ref{defaltra} produces an interesting and well-behaved algebraic invariant, we observe that the analogue of Proposition~\ref{compati} does not hold for support weights. In other words, while equivalent codes always have the same generalized weights, they might not have the same support weights. We illustrate this in the following example.
\end{remark}

\begin{example} \label{exnot}
Let $\mC$ be the binary code defined by
$$\mC:=\left\{ \begin{pmatrix} a & a \\ b & b \end{pmatrix}  \mid  a,b \in \F_2 \right\}.$$
Then $\mC$ is an optimal anticode of dimension 2. Therefore
$a_2(\mC)=1$. On the other hand, $\supp(\mC)=\F_2^2$, hence 
$cs_2(\mC)=2 \neq a_2(\mC)$.
Now observe that $\mC \sim \mC^t$.
In particular, $a_2(\mC)=a_2(\mC^t)=1$. However,
$cs_2(\mC)=2$, while $cs_2(\mC^t)=1$.
\end{example}

Generalized weights and support weights relate to each other as follows.

\begin{proposition}[\cite{U6}, Theorem 9]\label{pr:relate}
Let $\mC \subseteq \Mat$ be a non-zero code, and let $i \ge 1$ be an integer. If $m>n$, then $a_i(\mC)=cs_i(\mC)$.
If $m=n$, then $a_i(\mC) \le cs_i(\mC)$.
\end{proposition}

We stress that there exist codes $\mC \subseteq \Mat$ with $m=n$ and $a_i(\mC) < cs_i(\mC)$, e.g. the code $\mC$ of Example~\ref{exnot}.

\section{The $q$-Analogue of a Polymatroid}\label{secqPolym}

This section introduces $q$-polymatroids, that are a $q$-analogue of polymatroids. For more on (poly)matroids, see the standard references~\cite{Ox11,We76}.

\begin{definition} \label{def:qpoly}
A \textbf{$q$-polymatroid} is a pair $P=(\F_q^n,\rho)$ where $n \ge 1$ and $\rho$ is a function from the set of all subspaces of $\F_q^n$ to $\mathbb{R}$ such that, for all $A,B\subseteq \F_q^n$:
\begin{itemize} \setlength\itemsep{0em}
\item[(P1)] $0\leq \rho(A)\leq\dim(A)$,
\item[(P2)] if $A\subseteq B$, then $\rho(A)\leq \rho(B)$,
\item[(P3)] $\rho(A+B)+\rho(A\cap B)\leq \rho(A)+\rho(B)$.
\end{itemize}
\end{definition}

Notice that a $q$-polymatroid such that $\rho$ is integer-valued is a {\bf $q$-matroid} according to~\cite[Definition~2.1]{JP16}.

\begin{remark}
Our definition of $q$-polymatroid is slightly different from that of $(q,r)$-polymatroid given by Shiromoto in~\cite[Definition 2]{Sh18}. However, a $(q,r)$-polymatroid $(E,\rho)$ as defined by Shiromoto corresponds to the $q$-polymatroid $(E,\rho/r)$ according to our definition. Moreover, a $q$-polymatroid whose rank function takes values in $\mathbb{Q}$ corresponds to a $(q,r)$-polymatroid as defined by Shiromoto up to multiplying the rank function for an $r$ which clears denominators.
\end{remark}

\begin{remark}
One could also define a $q$-polymatroid $P$ as a pair $(\F_q^n,\rho)$ that satisfies $\rho(A)\geq 0$ for all $A\subseteq\F_q^n$, (P2), and (P3). 
Up to multiplying the rank function by a suitable constant, one may additionally assume that $\rho(x)\leq 1$ for all 1-dimensional subspaces $x\subseteq\F_q^n$. It is easy to show that this is equivalent to Definition~\ref{def:qpoly}.
\end{remark}

The definition of $q$-polymatroid that we propose is a direct $q$-analogue of the definition of an ordinary polymatroid, with the extra property that $\rho(A)\leq \dim(A)$ for all $A$. As in the ordinary case, a $q$-matroid is a $q$-polymatroid. At the end of Section~\ref{secStr} we give an example of a $q$-polymatroid that is not a $q$-matroid.
One has the following natural notion of equivalence for $q$-polymatroids.

\begin{definition} \label{defequipolym}
Let $(\F_q^n,\rho_1)$ and $(\F_q^n,\rho_2)$ be $q$-polymatroids. We say that $(\F_q^n,\rho_1)$ and $(\F_q^n,\rho_2)$ are \textbf{equivalent} if there exists an $\F_q$-linear isomorphism $\varphi:\F_q^n\to \F_q^n$ such that $\rho_1(A)=\rho_2(\varphi(A))$ for all $A\subseteq \F_q^n$. In this case we write $(\F_q^n,\rho_1) \sim (\F_q^n,\rho_2)$.
\end{definition}

We start by introducing a notion of duality for $q$-polymatroids.

\begin{definition} \label{defdualmatroid}
Let $P=(\F_q^n,\rho)$ be a $q$-polymatroid. For all subspaces $A \subseteq \F_q^n$ define
$$\rho^*(A)=\dim(A)-\rho(\F_q^n)+\rho(A^\perp),$$
where $A^\perp$ is the orthogonal complement of $A$ with respect to the standard inner product on $\F_q^n$.
We call $P^*=(\F_q^n,\rho^*)$ the \textbf{dual} of the $q$-polymatroid $P$.
\end{definition}

The proof of the next theorem is essentially the same as the proof of~\cite[Theorem~42]{JP16}.

\begin{theorem}
The dual $P^*$ is a $q$-polymatroid.
\end{theorem}
We will need the following property of dual $q$-polymatroids.

\begin{proposition} \label{passes}
Let $P_1=(\F_q^n,\rho_1)$ and $P_2=(\F_q^n,\rho_2)$ be $q$-polymatroids.
If $P_1 \sim P_2$, then $P_1^* \sim P_2^*$. 
Moreover, for every $q$-polymatroid $P$ we have $P^{**}=P$.
\end{proposition}

\begin{proof}
By definition, there exists an $\F_q$-isomorphism $\varphi: \F_q^n \to \F_q^n$ with the property that $\rho_1(A)=\rho_2(\varphi(A))$ for all $A \subseteq \F_q^n$. In particular, $\varphi(\F_q^n)=\F_q^n$. Therefore, by definition of $\rho_1^*$, for all $A \subseteq \F_q^n$ we have
$$\rho_1^*(A) = \dim(A) -\rho_2(\F_q^n) + \rho_2(\varphi(A^\perp)).$$
Now let $\psi: \F_q^n \to \F_q^n$ be the adjoint of $\varphi$ with respect to the standard inner product of $\F_q^n$. Then 
$\varphi$ is an $\F_q$-isomorphism, and 
$\varphi(A)^\perp=\psi(A)^\perp$ for all $A \subseteq \F_q^n$.
Therefore
$$\rho_1^*(A) = \dim(A) -\rho_2(\F_q^n) + \rho_2(\psi(A)^\perp)=\rho_2^*(\psi(A)).$$

If $P=(\F_q^n,\rho)$ is a $q$-polymatroid, then it is straighforward to check that $\rho^{**}(A)=\rho(A)$. This implies $P^{**}=P$.
\end{proof}

\section{Rank-Metric Codes and $q$-Polymatroids}\label{secCodesPolym}

Starting from a rank-metric code $\mC \subseteq \Mat$, in this section we construct two $q$-polymatroids: one associated to the column spaces, and the other to the row spaces. They will be denoted by $P(\mC,\cc)$ and $P(\mC,\rr)$ respectively. As the reader will see in the next sections, several structural properties of $\mC$ can be read off the associated $q$-polymatroids.

We start by studying subcodes of a given code, whose matrices are \textit{supported}  on a subspace $J \subseteq \F_q^n$ or $K \subseteq \F_q^m$. See~\cite{alblatt} for a lattice-theoretic definition of support.

\begin{notation}
Let $\mC \subseteq \Mat$ be a rank-metric code, and let $J\subseteq\mathbb{F}_q^n$ and $K \subseteq \F_q^m$ be subspaces. We define
  $$\mC(J,\cc) := \{M\in\mC \mid \colsp(M)\subseteq J\} \qquad \mbox{and} \qquad \mC(K,\rr) := \{M\in\mC \mid \rowsp(M)\subseteq K\},$$
where $\colsp(M)\subseteq\mathbb{F}_q^n$ and $\rowsp(M) \subseteq \F_q^m$ are the spaces generated over $\F_q$ by the columns, respectively the rows, of $M$. \end{notation}

Notice that $\mC(J,\cc)$ and $\mC(K,\rr)$ are subcodes of $\mC$ for all $J \subseteq \F_q^n$ and 
$K \subseteq \F_q^m$. 
In the sequel, we denote by $J^\perp$ the orthogonal of a space $J \subseteq \F_q^n$ with respect to the standard inner product of $\F_q^n$. We use the same notation for subspaces 
$K \subseteq \F_q^m$.
No confusion will arise with the trace-dual of a code $\mC \subseteq \Mat$.

%

\begin{notation}
Let $\mC \subseteq \Mat$ be a rank-metric code.
For subspaces $J \subseteq \F_q^n$ and $K \subseteq \F_q^m$ define the rational numbers
\begin{eqnarray*}
 \rho_\cc(\mC,J) &:=& (\dim(\mC) - \dim(\mC(J^\perp,\cc))/m, \\
 \rho_\rr(\mC,K) &:=& (\dim(\mC) - \dim(\mC(K^\perp,\rr))/n.
\end{eqnarray*} 
\end{notation}

For simplicity of notation, in the sequel we sometimes drop the index  $\mC$ and denote the rank functions simply by $\rho_\cc$ and $\rho_\rr$. 
The following result shows that a rank-metric code $\mC \subseteq \Mat$ gives rise to a pair of $q$-polymatroids via $\rho_\cc$ and $\rho_\rr$.

\begin{theorem} \label{givesrise}
Let $\mC \subseteq \Mat$ be a rank-metric code. The pairs $(\F_q^n,\rho_\cc)$ and $(\F_q^m,\rho_\rr)$ are $q$-polymatroids.
\end{theorem}

To prove the theorem we need a preliminary result, whose proof is  left to the reader.

\begin{lemma}\label{03-22-18}
Let $I, J \subseteq \F_q^n$ be subspaces. We have:
$$\Mat(I\cap J,\cc)=\Mat(I,\cc)\cap\Mat(J,\cc) \mbox{ and } \Mat(I + J,\cc)=\Mat(I,\cc)+\Mat(J,\cc).$$
\end{lemma}

\begin{proof}[Proof of Theorem~\ref{givesrise}]
We prove that $(\F_q^n,\rho_\cc)$ is a $q$-polymatroid. The proof that $(\F_q^m,\rho_\rr)$ is a $q$-polymatroid is completely analogous, hence we omit it.

We start by proving (P1). It is clear form the definition that $\rho_\cc(J)\geq0$. The other inequality follows from~\cite[Lemma~28]{RR15}:
$$\rho_\cc(J)=(\dim(\mC)-(\dim(\mC)-m(n-\dim J^\perp)+\dim(\mC^\perp(J))))/m \leq \dim(J).$$
Now let $I,J\subseteq\F_q^n$ such that $I\subseteq J$. Then $\mC(J^\perp,\cc)\subseteq \mC(I^\perp,\cc)$, thus $\rho_\cc(I)\leq \rho_\cc(J)$. This establishes (P2). For (P3), we have
\begin{eqnarray*}
&&\dim \mC((I+J)^\perp,\cc)+\dim \mC((I\cap J)^\perp,\cc)\\
&=&\dim (\mC\cap \Mat(I^\perp\cap J^\perp,\cc))+\dim (\mC\cap \Mat(I^\perp+ J^\perp,\cc))\\
&\geq& \dim (\mC\cap \Mat(I^\perp,\cc)\cap\Mat( J^\perp,\cc))+
\dim ((\mC\cap (\Mat(I^\perp,\cc))+(\mC\cap \Mat(J^\perp,\cc)))\\
&=& \dim (\mC\cap \Mat(I^\perp,\cc))+
\dim (\mC\cap\Mat( J^\perp,\cc)),
\end{eqnarray*}
where the first equality follows from~\cite[Lemma~27]{RR15}. The inequality follows from combining Lemma~\ref{03-22-18} with $\mC\cap (\Mat(I^\perp,\cc)+\Mat(J^\perp,\cc))\supseteq
(\mC\cap \Mat(I^\perp,\cc))+(\mC\cap\Mat(J^\perp,\cc)).$
\end{proof}

\begin{notation}
The $q$-polymatroids associated to a rank-metric code $\mC \subseteq \Mat$ are denoted by $P(\mC,\cc)$ and $P(\mC,\rr)$, respectively. 
\end{notation}

\section{Structural Properties of Codes via $q$-Polymatroids} \label{secStr}

In this section we investigate some connections between rank-metric codes and the associated $q$-polymatroids.
We show that the $q$-polymatroids associated to a code $\mC$ determine the dimension of the code and its minimum distance, and characterize the property of being MRD. 

\begin{proposition}
Let $\mC \subseteq \Mat$ be a rank-metric code. Then $$\dim(\mC)=m\cdot\rho_\cc(\mC,\F_q^n)=n\cdot\rho_\rr(\mC,\F_q^m).$$
\end{proposition}

The above result follows directly from the definitions.
We now relate the minimum distance of a code with the rank functions of the associated $q$-polymatroids.

\begin{proposition} \label{pp1}
Let $\mC \subseteq \Mat$ be a non-zero rank-metric code. The following are equivalent:
\begin{enumerate}[label=(\arabic*)]  \setlength\itemsep{0em}
\item $d(\mC) \ge d$,
\item $\rho_\cc(J)= \dim(\mC)/m$ for all $J \subseteq \F_q^n$ with $\dim(J) \ge n-d+1$,
\item $\rho_\rr(K)= \dim(\mC)/n$ for all $K \subseteq \F_q^m$ with $\dim(K) \ge m-d+1$.
\end{enumerate}
\end{proposition}
\begin{proof}
It is easy to see that the following are equivalent:
\begin{enumerate}[label=(\arabic*$'$)] \setlength\itemsep{0em}
\item $d(\mC) \ge d$,
\item $\mC(J,\cc) = \{0\}$ for all $J \subseteq \F_q^n$ with $\dim(J) \le d-1$,
\item $\mC(K,\rr) = \{0\}$ for all $K \subseteq \F_q^m$ with $\dim(K) \le d-1$.
\end{enumerate}
By definition, for all $J \subseteq \F_q^n$ and $K \subseteq \F_q^m$ we have
$m \rho_\cc(J) =\dim(\mC) - \dim(\mC(J^\perp,\cc))$ and $n \rho_\rr(K) =\dim(\mC) - \dim(\mC(K^\perp,\rr))$.
Hence $(2) \Leftrightarrow (2')$ and $(3) \Leftrightarrow (3')$.
\end{proof}

Therefore, the minimum distance of a rank-metric code can be expressed in terms of the rank function of one of the associated $q$-polymatroids as follows.

\begin{corollary}
Let $0\neq \mC \subseteq \Mat$. The minimum distance of $\mC$ is
\begin{eqnarray*}
d(\mC) &=& n+1-\min\left\{d \mid \rho_\cc(J)=\frac{\dim(\mC)}{m}\mbox{ for all $J \subseteq \F_q^n$ with $\dim(J)=d$} \right\} \\
&=&m+1-\min\left\{d \mid  \rho_\rr(K)=\frac{\dim(\mC)}{n}\mbox{ for all $K \subseteq \F_q^m$ with $\dim(K)=d$} \right\}.
\end{eqnarray*}
\end{corollary}

This allows us to characterize the property of being MRD in terms of the rank function of one of the associated $q$-polymatroids.

\begin{theorem} \label{tt1}
Let $\mC \subseteq \Mat$ be a non-zero code of minimum distance $d$. The following are equivalent:
\begin{enumerate}[label=(\arabic*)]  \setlength\itemsep{0em}
\item $\mC$ is MRD,
\item $\rho_\cc(J)=\dim(J)$ for all $J \subseteq \F_q^n$ with $\dim(J) \le n-d+1$,
\item $\rho_\cc(J)=\dim(J)$ for some $J \subseteq \F_q^n$ with $\dim(J) = n-d+1$.
\end{enumerate}
\end{theorem}
\begin{proof}
Assume that $\mC$ is MRD. We claim that $$\dim(\mC(J,\cc))=\dim(\mC)-m(n-\dim(J)) \mbox{ for all $J \subseteq \F_q^n$ with $\dim(J) \ge d-1.$}$$
This is straightforward if $\dim(J)=d-1$. When $\dim(J) \ge d$, it follows from~\cite[Lemma 48]{alblatt}. Let $J \subseteq \F_q^n$ be a subspace with $\dim(J) \le n-d+1$. Since $\dim(J^\perp) \le d-1$ and $\dim(\mC(J,\cc))=\dim(\mC)-m(n-\dim(J))$ we obtain
$$m\rho_\cc(J)=\dim(\mC)-\dim(\mC(J^\perp,\cc)) = \dim(\mC) - \dim(\mC) +m  \dim(J) =m\dim(J).$$
This establishes $(1)  \Rightarrow (2)$.

It is clear that $(2)$ implies $(3)$. So we assume that (3) holds and prove $(1)$. 
Since $\dim(J)=n-d+1$, then $\dim(J^\perp)=d-1$, therefore $\dim(\mC(J^\perp,\cc))=0$. It follows that
$$m\dim(J) = m\rho_\cc(J) =\dim(\mC) - \dim(\mC(J^\perp,\cc))=\dim(\mC),$$
from which we obtain $\dim(\mC)=m\dim(J)=m(n-d+1)$. Hence $\mC$ is MRD.
\end{proof}

\begin{remark}
If $m=n$ and $0\neq\mC\subseteq \Mat$, then the same proof as in Theorem~\ref{tt1} shows that the following are equivalent:
\begin{itemize} \setlength\itemsep{0em}
\item $\mC$ is MRD, 
\item $\rho_\cc(K)=\dim(K)$ for all $K\subseteq \F_q^m$ with $\dim(K) = m-d+1$
\item $\rho_\cc(K)=\dim(K)$ for some $K\subseteq \F_q^m$ with $\dim(K) = m-d+1$. 
\end{itemize}
\end{remark}

\bigskip

Combining Proposition~\ref{pp1} and Theorem~\ref{tt1} we obtain an explicit formula for the rank function of the (column) $q$-polymatroid associated to an MRD code.

\begin{corollary} \label{formulaMRD}
Let $\mC \subseteq \Mat$ be a non-zero MRD code of minimum distance $d$. Then for all $J \subseteq \F_q^n$ we have
\begin{equation} \label{foor}
\rho_\cc(J)= \left\{ \begin{array}{ll} n-d+1 & \mbox{ if } \dim(J) \ge n-d+1, \\ \dim(J) & \mbox{ if } \dim(J) \le n-d+1.\end{array}\right.
\end{equation}
\end{corollary}

In particular, the $q$-polymatroid associated to an MRD code has an integer-valued rank function, i.e., it is a $q$-matroid. It is in fact the uniform $q$-matroid, as explained in~\cite[Example~4.16]{JP16}

It is natural to expect that equivalent rank-metric codes give rise to equivalent $q$-polymatroids. This is true in the following precise sense.

\begin{proposition} \label{equivpoly}
Let $\mC_1, \mC_2 \subseteq \Mat$ be rank-metric codes. Assume $\mC_1 \sim \mC_2$. If $m>n$, then 
$P(\mC_1,\cc) \sim P(\mC_2,\cc)$ and $P(\mC_1,\rr) \sim P(\mC_2,\rr)$. If $n=m$, then one of the following holds:
\begin{itemize}\setlength\itemsep{0em}
\item $P(\mC_1,\cc) \sim P(\mC_2,\cc)$ and $P(\mC_1,\rr) \sim P(\mC_2,\rr)$,
\item  $P(\mC_1,\cc) \sim P(\mC_2,\rr)$ and $P(\mC_1,\rr) \sim P(\mC_2,\cc)$.
\end{itemize}
\end{proposition}

\begin{proof}
Since $\mC_1\sim\mC_2$, then either $\mC_2=A\mC_1 B$ for some invertible $A,B$, or $\mC_2=A\mC_1^t B$ for some invertible $A,B$ and $m=n$.
Since the proofs are similar, we only treat the case when there exists invertible matrices $A,B$ such that $\mC_2=A\mC_1 B$. 
Let $\psi: \F_q^n \to \F_q^n$ be the $\F_q$-linear isomorphism associated to the matrix $A$ with respect to the standard basis. Fix a subspace $J \subseteq \F_q^n$. Multiplication by $A$ on the left and $B$ on the right induces a bijection
\begin{equation} \label{uno}
\mC_1(J^\perp,\cc) \to \mC_2(\psi(J^\perp),\cc).
\end{equation}
Let $\varphi: \F_q^n \to \F_q^n$ denote the $\F_q$-linear isomorphism associated to the matrix $(A^{-1})^t$ with respect to the standard basis. Then we have
$\psi(J^\perp)=\varphi(J)^\perp$, hence bijection (\ref{uno}) can be thought of as a bijection
\begin{equation} \label{due}
\mC_1(J^\perp,\cc) \to \mC_2(\varphi(J)^\perp,\cc).
\end{equation}
Therefore for all subspaces $J \subseteq \F_q^n$ we have 
$\rho_\cc(\mC_1,J)= \rho_\cc(\mC_2,\varphi(J))$. This establishes the 
$q$-polymatroid equivalence $P(\mC_1,\cc) \sim P(\mC_2,\cc)$.
The equivalence $P(\mC_1,\rr) \sim P(\mC_2,\rr)$ can be shown similarly.
\end{proof}

Proposition~\ref{equivpoly} says that equivalent codes have equivalent associated $q$-polymatroids. The next example shows that the converse is false in general, i.e., that inequivalent codes may have equivalent (in fact, even identical) associated
$q$-polymatroids. 

\begin{example}
Let $q=2$ and $m=n=4$. Let $\mC_1$ be the code of~\cite[Example 7.2]{byra}, i.e., the code generated by the four linearly independent binary matrices
$$\begin{pmatrix}
   1 & 0 & 0 & 0 \\ 0 & 0 & 0 & 1 \\ 0 & 0 & 1 & 0 \\ 0 & 1 & 0 & 0
  \end{pmatrix}, \ \ \ \ \ 
\begin{pmatrix}
   0 & 1 & 0 & 0 \\ 0 & 0 & 1 & 1 \\ 0 & 0 & 0 & 1 \\ 1 & 1 & 0 & 0
  \end{pmatrix},  \ \ \ \ \ 
\begin{pmatrix}
   0 & 0 & 1 & 0 \\ 0 & 1 & 1 & 1 \\ 1 & 0 & 1 & 0 \\ 1 & 0 & 0 & 1
  \end{pmatrix},\ \ \ \ \ 
\begin{pmatrix}
   0 & 0 & 0 & 1 \\ 1 & 1 & 1 & 0 \\ 0 & 1 & 0 & 1 \\ 0 & 1 & 1 & 1
  \end{pmatrix}.
$$  
The code $\mC_1$ is MRD and has minimum distance $d(\mC_1)=4$. Let $\mC_2$ be a rank-metric code obtained 
from a Gabidulin code $C \subseteq \F_{2^4}^4$ of minimum distance $4$ via Proposition~\ref{lift}. By~\cite[Example 7.2]{byra}, the code $\mC_1$ has covering radius $\textnormal{cov}(\mC_1)=2$, while it is well known that $\textnormal{cov}(\mC_2)=d(\mC)-1=3$. Since the covering radius of a code is preserved under isometries, we conclude that the codes $\mC_1$ and $\mC_2$ are not equivalent.

On the other hand, the four codes $\mC_1$, $\mC_2$, $\mC_1^t$ and $\mC_2^t$ are all MRD with the same parameters. Therefore by Corollary~\ref{formulaMRD} the rank function of their $q$-polymatroids is determined and given by the formula in (\ref{foor}). This shows that $P(\mC_1,\cc)=P(\mC_1,\rr) = P(\mC_2,\cc)=P(\mC_2,\rr)$, although $\mC_1 \not\sim \mC_2$.
\end{example}

It is known from~\cite{JP16} that a vector rank-metric code $C \subseteq \F_{q^m}^n$ gives rise to a $q$-matroid $M(C)$ on $\F_q^n$. In our notation, $M(C)=P(\Gamma(C),\cc)$,
where $\Gamma$ is any $\F_q$-basis of $\F_{q^m}$. 

\begin{proposition}\label{propequal}
Let  $C \subseteq \F_{q^m}^n$ be a vector rank-metric code, and let $\Gamma,\Gamma'$ be $\F_q$-bases of $\F_{q^m}$.
We have $P(\Gamma(C),\cc) = P(\Gamma'(C),\cc)$ and $P(\Gamma(C),\rr) \sim P(\Gamma'(C),\rr)$.
\end{proposition} 

\begin{proof}
The statement that $P(\Gamma(C),\cc) = P(\Gamma'(C),\cc)$ follows from~\cite[Corollary~4.7]{JP16}. The statement that $P(\Gamma(C),\rr) \sim P(\Gamma'(C),\rr)$ follows by Propositions~\ref{equiv} and~\ref{equivpoly}.
\end{proof}

We continue by showing that there exist rank-metric codes whose associated $q$-polymatroids are not $q$-matroids. 
Even more, in the next example we show that there are $q$-polymatroids such that no non-zero multiple of their rank function defines a $q$-matroid.

\begin{example}
Let $q=3$ and $n=m=2.$ Let $\mC$ be a rank-metric code generated by the matrices
$$M_1:=\begin{pmatrix}1 & 0 \\0 & 0\end{pmatrix},\qquad
M_2:=\begin{pmatrix}0 & 1 \\0 & 0\end{pmatrix},\qquad
M_3:=\begin{pmatrix}0 & 0 \\1 & 0\end{pmatrix}.$$
Consider the subspaces $J:=\left<(1,0)\right>$ and $I:=\left<(0,1)\right>$.
Since $\mC(J^\perp,\cc)=\left< M_3\right>$, we have
$\rho_\cc(\mC,J)=1.$ As $\mC(I^\perp,\cc)=\left<M_1,M_2\right>,$ then $ \rho_\cc(\mC,I)=1/2$. Hence $P(\mC,\cc)$ is a $q$-polymatroid which
 is not a $q$-matroid. 

Let $\alpha\in\mathbb{R}$  with $\alpha \neq 0$, and consider the function $\rho:=\alpha\rho_\cc$. Since $\rho(J)=\alpha$, in order for $\rho$ to be the rank function of a $q$-polymatroid it must be $0<\alpha\leq 1$. Then $\rho(I)=\alpha/2$ is not an integer, so $\rho$ cannot be the rank function of a $q$-matroid.
\end{example}

\section{Generalized Weights as $q$-Polymatroid Invariants}\label{secGenWtPoly}

In this section, we provide further evidence that the $q$-polymatroids associated to a rank-metric code adequately capture the structure of the code. More precisely, in the next theorem we show that the generalized rank-weights of the code are an invariant of the associated $q$-polymatroids. Later in the section, we show that the property of being an optimal anticode can be characterized in terms of the rank function of the associated $q$-polymatroids.

\begin{theorem}\label{03-23-18}
Let $\mC \subseteq \Mat$ be a non-zero rank-metric code and let $1 \le i \le \dim(\mC)$ be an integer. If $n>m$ we have
$$a_i(\mC) = \min \{n-\dim(J) \mid J \subseteq \F_q^n, \ \dim(\mC)-m  \rho_\cc(\mC,J) \ge i\}.$$
If $n=m$ we have $$a_i(\mC)=\min\{a_i(\mC,\cc), \ a_i(\mC,\rr)\},$$
where
\begin{eqnarray*}
a_i(\mC,\cc) &:=& \min \{n-\dim(J) \mid J \subseteq \F_q^n, \ \dim(\mC)-m  \rho_\cc(\mC,J) \ge i\}, \\
a_i(\mC,\rr) &:=& \min \{m-\dim(K) \mid K \subseteq \F_q^m, \ \dim(\mC)-n  \rho_\rr(\mC,K) \ge i\}. \\
\end{eqnarray*}
\end{theorem}

\begin{proof}
Let $J\subseteq \F_q^n$, then by~\cite[Lemma~26]{RR15}
\begin{equation}\label{equalitydim}
\dim(\Mat(J^\perp,\cc))=m\dim(J^\perp)=m(n-\dim(J)).
\end{equation}

Assume that $m>n$. By~\cite[Theorem~6]{pazzis}, the optimal anticodes in $\Mat$ are the spaces of the form $\Mat(J^\perp, \cc)$, where $J$ ranges over the subspaces of $\F_q^n$. 
Therefore 
\begin{eqnarray*}
m  \cdot a_i(\mC)&=& \min\{\dim(\Mat(J^\perp,\cc)) \mid J \subseteq \F_q^n, \ \dim(\mC \cap \Mat(J^\perp,\cc)) \ge i\}\\
&=& m \cdot \min \{n-\dim(J) \mid J \subseteq \F_q^n, \ \dim(\mC)-m  \rho_\cc(\mC,J) \ge i\},
\end{eqnarray*}
where the last equality follows from (\ref{equalitydim}) and the definition of $\rho_\cc(\mC,J)$. 

Now assume that $n=m$. By~\cite[Theorem~4]{pazzis}, the anticodes in $\Mat$ are the spaces of the form
$\Mat(J^\perp, \cc)$ or $\Mat(J^\perp, \rr)$, as $J$ ranges over the subspaces of $\F_q^n.$ 
Then \begin{eqnarray*}
a_i(\mC)&=& \frac{1}{n} \min\left\{\dim(\Mat(J^\perp,\cc)) \mid J \subseteq \F_q^n, \ \dim(\mC \cap \Mat(J^\perp,\cc)) \ge i\right\}\cup\\
&&\qquad \qquad \qquad \cup \left\{\dim(\Mat(J^\perp,\rr)) \mid J\subseteq \F_q^n, \ \dim(\mC \cap \Mat(J^\perp,\rr)) \ge i\right\}\\
&=& \min\{a_i(\mC,\cc), \ a_i(\mC,\rr)\},
\end{eqnarray*}
where the last equality follows from (\ref{equalitydim}) and the definition of $\rho_\cc(\mC,J), \rho_\rr(\mC,J)$.
\end{proof}

In the next theorem, we prove that the property of being an optimal anticode is captured by the rank function of the associated $q$-polymatroids. 

\begin{theorem} \label{rhoanti}
Let $\mathcal{C} \subseteq \Mat$ be a rank-metric code and let $t=\mbox{maxrk}(\mC)$. The following are equivalent:
\begin{enumerate}
\item $\mC$ is an optimal anticode,
\item $\left\{\rho_\cc(\mC,J)\mid J \subseteq \F_q^n\right\}=\left\{0,1,\ldots,t\right\}$, or $\left\{\rho_\rr(\mC,J)\mid J \subseteq \F_q^n\right\}=\left\{0,1,\ldots,t\right\}$ and $m=n$,
\item $\rho_\cc(\mC,\F_q^n)=t$, or $\rho_\rr(\mC,\F_q^n)=t$ and $m=n$.
\end{enumerate}
In particular, the $q$-polymatroid associated to an optimal anticode is a $q$-matroid.
\end{theorem}

\begin{proof}
$(1)\Rightarrow(2)$ By~\cite[Theorems~4 and~6]{pazzis},
either $\mC=\Mat(K,\cc)$ for a $t$-dimensional subspace $K \subseteq \F_q^n$, or $\mC=\Mat(K,\rr)$ for a $t$-dimensional subspace $K \subseteq \F_q^m$, where the latter is only possible if $m=n$. We assume that $\mC=\Mat(K,\cc)$, as the proof in the other situation is analogous. One has, for all $J \subseteq \F_q^n$,
$$\rho_\cc(\mC,J) = (mt - \dim(\Mat(K,\cc)\cap \Mat(J^\perp,\cc))/m = t-\dim(K \cap J^\perp),$$
where the second equality follows from Lemma~\ref{03-22-18} and~\cite[Lemma~26]{RR15}.
Hence we obtain 
$$\left\{\rho_\cc(\mC,J)\mid J \subseteq \F_q^n\right\}=\left\{0,1,\ldots,t\right\} \mbox{ if } \mC=\Mat(K,\cc),$$ 
$$\left\{\rho_\rr(\mC,J)\mid J \subseteq \F_q^n\right\}=\left\{0,1,\ldots,t\right\}\mbox{ if }\mC=\Mat(K,\rr).$$

\noindent $(3)\Rightarrow(1)$ We have $$\rho_\cc(\mC,\F_q^n)=\dim(\mC)/m \mbox{ and } \max\{\rho_\rr(\mC,K)\mid K\subseteq\F_q^m\}=\rho_\rr(\mC,\F_q^m)=\dim(\mC)/n.$$
Then either $\dim(\mC)/m$, or $\dim(\mC)/n=t$ and $m=n$. Either way one has $\dim(\mC)=mt$, hence $\mC$ is an optimal anticode.
\end{proof}

\begin{corollary}\label{formulaoptimalanticode}
Let $\mathcal{C} \subseteq \Mat$ be an optimal anticode and let $t=\mbox{maxrk}(\mC)$. If $m>n$, then $P(\mC,\cc)\sim(\F_q^n,\rho)$ where 
\begin{equation}\label{matroidanticodes}
\rho(J)=\dim(J+\langle e_1,\ldots,e_{n-t}\rangle)-(n-t)
\end{equation} 
and $e_i$ denotes the $i$-th vector of the standard basis of $\F_q^n$.
If $m=n$, then either $P(\mC,\cc)\sim(\F_q^n,\rho)$ or $P(\mC,\rr)\sim(\F_q^n,\rho)$.
\end{corollary}

\begin{proof}
If $m>n$, then $\mC=\Mat(K,\cc)$ for some $K\subseteq\F_q^n$ of $\dim(K)=t$. If $m=n$, then either $\mC=\Mat(K,\cc)$ or $\mC^t=\Mat(K,\cc)$, for some $K\subseteq\F_q^n$ of $\dim(K)=t$. Since $P(\mC^t,\cc)=P(\mC,\rr)$, it suffices to consider the case when $\mC=\Mat(K,\cc)$. Up to code equivalence, we may also assume without loss of generality that $K=\langle e_{n-t+1},\ldots,e_n\rangle$.

It follows from the proof of Theorem~\ref{rhoanti} that $\rho_\cc(\mC,J)=t-\dim(K \cap J^\perp)$. Therefore
$\rho_\cc(\mC,J)=t-(n-dim(\langle e_{n-t+1},\ldots,e_n\rangle \cap J^\perp)^\perp)=\dim(J+\langle e_1,\ldots,e_{n-t}\rangle)-(n-t).$\end{proof}

\begin{remark}
One consequence of our results is that, in certain cases, the generalized weights of a code determine the associated $q$-polymatroid $P(\mC,\cc)$ up to equivalence. 
This is the case e.g. in the following situations:\begin{itemize} 
\item if $\mC$ has the generalized weights of an MRD code, then $\mC$ is MRD and $P(\mC,\cc)$ is the uniform $q$-matroid (see Corollary~\ref{formulaMRD}),
\item if $\mC$ has the generalized weights of an optimal anticode, then $\mC$ is an optimal anticode and $P(\mC,\cc)$ is the $q$-matroid described in Corollary~\ref{formulaoptimalanticode},
\item if $\dim(\mC)=1$, then $\mC=\langle M\rangle$ and $a_1(\mC)=d_{\min}(\mC)=\rk(M)$. Moreover, $P(\mC,\cc)$ is given by $$\rho_\cc(\mC,J)=\left\{\begin{array}{ll}
0 & \mbox{if $\colsp(M)\subseteq J^\perp$,}\\
\frac{1}{m} & \mbox{else.}
\end{array}\right.$$
Notice that if $\mC_1=\langle M_1\rangle$ and $\mC_2=\langle M_2\rangle$ have the same minimum distance, then $P(\mC_1,\cc)\sim P(\mC_2,\cc)$. In fact $\rho_\cc(\mC_1,J)=\rho_\cc(\mC_2,\varphi(J))$, where $\varphi:\F_q^n\rightarrow\F_q^n$ is an $\F_q$-linear isomorphism such that $\varphi(\colsp(M_1))=\colsp(M_2)$.
\end{itemize}
One should however not expect this to be the case in general. In other words, the generalized weights of a rank-metric code $\mC$ are invariants of the associated $q$-polymatroid $P(\mC,\cc)$, but they do not determine it, as the next example shows. Similar examples may be found for rectangular matrices.
\end{remark}

\begin{example}
Let $\mC_1,\mC_2\subseteq\Mat_{2\times 2}(\F_2)$, 
$$\mC_1=\left\langle\begin{pmatrix} 1 & 0 \\ 0 & 1 \end{pmatrix}, \begin{pmatrix} 0 & 1 \\ 0 & 0 \end{pmatrix}\right\rangle, \;\;\;
\mC_2=\left\langle\begin{pmatrix} 0 & 1 \\ 1 & 0 \end{pmatrix}, \begin{pmatrix} 0 & 1 \\ 0 & 0 \end{pmatrix}\right\rangle.$$
The codes $\mC,\mD$ have generalized weights $a_1(\mC_1)=a_1(\mC_2)=1$ and $a_2(\mC_1)=a_2(\mC_2)=2$. In fact, any rank-metric code of dimension $2$ and minimum distance $1$ which is not an optimal anticode has the same generalized weights as $\mC_1$ and $\mC_2$.

Let $P(\mC_1,\cc)=(\F_2^2,\rho_1)$ and $P(\mC_2,\cc)=(\F_2^2,\rho_2)$. Let $J\subseteq\F_2^2$ be a 1-dimensional linear subspace. Then 
$$\rho_1(J)=\left\{\begin{array}{ll}
\frac{1}{2} & \mbox{if $J=\langle(0,1)\rangle$}\\
1 & \mbox{if $J=\langle(1,0)\rangle$ or $J=\langle(1,1)\rangle$}
\end{array}\right.$$
while
$$\rho_2(J)=\left\{\begin{array}{ll}
\frac{1}{2} & \mbox{if $J=\langle(0,1)\rangle$ or $J=\langle(1,0)\rangle$,}\\
1 & \mbox{if $J=\langle(1,1)\rangle$}
\end{array}\right.$$
Therefore $P(\mC_1,\cc)\not\sim P(\mC_2,\cc).$
Notice moreover that $P(\mC_1,\cc)\sim P(\mC_1,\rr)$ and $P(\mC_2,\cc)\sim P(\mC_2,\rr)$. 
\end{example}

\section{Duality}\label{secDuality}

In this last section of the paper we establish a connection between the notions of code duality and $q$-polymatroid duality.
We start by showing that the $q$-polymatroids associated to the dual code $\mC^\perp$ are the duals of the $q$-polymatroids associated to the 
original code $\mC$.

\begin{theorem} \label{mainduality}
Let $\mathcal{C} \subseteq \Mat$ be a rank-metric code. We have $P(\mathcal{C},\cc)^*=P(\mathcal{C}^\perp,\cc)$ and 
$P(\mathcal{C},\rr)^*=P(\mathcal{C}^\perp,\rr)$. 
\end{theorem}

\begin{proof}
We only show the result for $P(\mathcal{C},\cc)$. The proof for $P(\mathcal{C},\rr)$ is analogous.
Let $J\subseteq \F_q^n$ be a subspace.
Since $\rho_\cc(\mC,J) = (\dim(\mC) - \dim(\mC(J^\perp,\cc)))/m$, then
$$\rho^*_\cc(\mC,J)=\dim(J) - \dim (\mC)/m + (\dim (\mC)-\dim(\mC(J,\cc)))/m=\dim(J) - \dim(\mC(J,\cc))/m.$$
Therefore by~\cite[Lemma~28]{RR15} one has
\begin{equation*}
m\rho^*_\cc(\mC,J)-m\rho_\cc(\mC^\perp,J)=m\dim(J)-\dim(\mC^\perp)-\dim(\mC)+mn-m\dim(J)=0.\qedhere
\end{equation*}
\end{proof}

Finally, it is natural to ask how the $q$-polymatroids associated to the dual of a vector rank-metric code relate to the $q$-polymatroids associated to the original vector rank-metric code. 
It turns out that they are dual to each other, as the following result shows.

\begin{corollary}
Let $C \subseteq \F_{q^m}^n$ be a vector rank-metric code, and let $\Gamma$ be a basis of  $\F_{q^m}$ over $\F_q$.
We have
$$P(\Gamma(C^{\dbot}),\cc) = P(\Gamma^*(C),\cc)^*=P(\Gamma(C),\cc)^* \mbox{ and } P(\Gamma(C^{\dbot}),\rr) = P(\Gamma^*(C),\rr)^*\sim P(\Gamma(C),\rr)^*$$
where $\Gamma^*$ is the dual of the basis $\Gamma$.
\end{corollary}

\begin{proof}
Applying~\cite[Theorem~21]{RR15} to $C$ we obtain $\Gamma(C^{\dbot})=\Gamma^*(C)^\perp$, hence $P(\Gamma(C^{\dbot}),\cc) = P(\Gamma^*(C)^\perp,\cc)$ and $P(\Gamma(C^{\dbot}),\rr) = P(\Gamma^*(C)^\perp,\rr)$.  
On the other hand, Theorem~\ref{mainduality} gives $P(\Gamma^*(C)^\perp,\cc)=P(\Gamma^*(C),\cc)^*$ and $P(\Gamma^*(C)^\perp,\rr)=P(\Gamma^*(C),\rr)^*$. 
By Proposition~\ref{propequal} we have $P(\Gamma^*(C),\cc)=P(\Gamma(C),\cc)$ and $P(\Gamma^*(C),\rr)\sim P(\Gamma(C),\rr)$.
Therefore by Proposition~\ref{passes} it follows that $P(\Gamma^*(C),\cc)^*=P(\Gamma(C),\cc)^*$ and $P(\Gamma^*(C),\rr)^*\sim P(\Gamma(C),\rr)^*$. 
\end{proof}

%
%
%
%


\begin{thebibliography}{99}

\bibitem{Berger} T.P. Berger,
\emph{Isometries for rank distance and permutation group of Gabidulin codes},
IEEE Transactions on Information Theory, {\bf 49} (2002), no. 11, 3016--3019.

\bibitem{byra} E. Byrne and A. Ravagnani, 
\emph{Covering radius of matrix codes endowed with the rank metric}, 
SIAM Journal on Discrete Mathematics, {\bf 31} (2017), no. 2, 927--944.

\bibitem{Cr64} H. Crapo,
\emph{On the theory of combinatorial independence},
Ph.D. Thesis, Massachusetts Institute of Technology, Dept. of Mathematics (1964).

\bibitem{del1} P. Delsarte,
\emph{Bilinear forms over a finite field, with applications to coding theory},
Journal of Combinatorial Theory, Series A, {\bf 25} (1978), no. 3, 226--241.

\bibitem{gabid} E. Gabidulin,
\emph{Theory of codes with maximum rank distance},
Problems of Information Transmission, {\bf 1} (1985), no. 2, 1--12.

\bibitem{costch} E. Gorla, A. Ravagnani,
\emph{Codes endowed with the rank metric}, in
``Random Network Coding and Designs Over $GF(q)$'',
Eds. M. Greferath, M. Pavcevic, A. Vazquez-Castro, N. Silberstein,
Signals and Communication Technology, Springer-Verlag Berlin (2018).

\bibitem{hua} L.-K. Hua, 
\emph{A theorem on matrices over a sfield and its applications},
Acta Mathematica Sinica, {\bf 1} (1951), 109--163.

\bibitem{Ondef} R. Jurrius, R. Pellikaan,
\emph{On defining generalized rank weights}, 
Advances in Mathematics of Communications, {\bf 11} (2017), no. 1, 225--235.

\bibitem{JP16} R. Jurrius, R. Pellikaan,
\emph{Defining the $q$-analogue of a matroid},
Electronic Journal of Combinatorics, {\bf 25} (2018), no. 3, 1--32.

\bibitem{U6} U. Mart{\'\i}nez-Pe{\~n}as, R. Matsumoto, 
\emph{Relative generalized matrix weights of matrix codes for universal security on wire-tap networks},
IEEE Transactions on Information Theory, {\bf 64} (2018), no. 4, 2529--2549.

\bibitem{Klove} T. Kl{\o}ve,
\emph{The weight distribution of linear codes over $GF(q^l)$ having generator matrix over $GF(q)$},
Discrete Mathematics, {\bf 23} (1978), 159--168.

\bibitem{KMU} J. Kurihara, R. Matsumoto, T. Uyematsu,
\emph{Relative generalized rank weight of linear codes and its applications to network coding},
IEEE Transactions on Information Theory, {\bf 61} (2015), no. 7, 3912--3936.

\bibitem{OS} F. Oggier, A. Sboui,
\emph{On the existence of generalized rank weights}, in 
``IEEE ISIT-2012, International Symposium on Information Theory'' (2012), 406--410.

\bibitem{Ox11} J. Oxley,
\emph{Matroid Theory},
Oxford Graduate Texts in Mathematics, Second Edition (2011).

\bibitem{OW93} J. Oxley, G. Whittle,
\emph{A characterization of Tutte invariants of $2$-polymatroids},
Journal of Combinatorial Theory, Series B, {\bf 59} (1993), no. 2, 210--244.

\bibitem{RR15} A. Ravagnani,
\emph{Rank-metric codes and their duality theory}.
Designs, Codes and Cryptography,  {\bf 80} (2016), no. 1, 197--216.

\bibitem{albgen} A. Ravagnani,
\emph{Generalized weights: an anticode approach},
Journal of Pure and Applied Algebra, {\bf 220} (2016), no. 5, 1946--1962.

\bibitem{alblatt} A. Ravagnani, 
\emph{Duality of Codes Supported on Regular Lattices, with an Application to Enumerative Combinatorics},
Designs, Codes and Cryptography, {\bf 86} (2018), no.9, 2035--2063.

\bibitem{roth} R. M. Roth, 
\emph{Maximum-Rank Array Codes and their Application to Criss-cross Error Correction},
IEEE Transactions on Information Theory, {\bf 37} (1991), no. 2, 328--336.

\bibitem{pazzis} C. de Seguins Pazzis,
\emph{The classification of large spaces of matrices with bounded rank},
Israel Journal of Mathematics, {\bf 208} (2015), no. 1, 219--259.

\bibitem{Sh18} K. Shiromoto,
\emph{Codes with the rank metric and matroids},
to appear in Designs, Codes and Cryptography.

\bibitem{wan} Z.-X. Wan, 
\emph{A proof of the automorphisms of linear groups over a sfield of characteristic 2},
Scientia Sinica {\bf 11} (1962), 1183--1194.
 
\bibitem{wan2} Z.-X. Wan, 
\emph{Geometry of matrices. In memory of Professor L. K. Hua (1910--1985)}, 
World Scientific, Singapore (1996).

\bibitem{Wei} V. Wei,
\emph{Generalized hamming weights for linear codes},
IEEE Transactions on Information Theory, {\bf 37} (1991), no. 5, 1412--1418.

\bibitem{We76} D. Welsh,
\emph{Matroid theory},
Dover Publications, INC. Mineola New York (1976).

\end{thebibliography}
\end{document}